%% file: iclr_3198_arXiv.tex
\documentclass{article}
\usepackage{iclr2027_conference,times}
\input{math_commands.tex}

\usepackage{hyperref}
\usepackage{url}
\usepackage{amsmath,amssymb,amsthm,booktabs,graphicx}
\newtheorem{proposition}{Proposition}
\newtheorem{theorem}{Theorem}
\newtheorem{corollary}{Corollary}
 \iclrfinalcopy  

\author{Daria Leshchikova \\
Fleamily, Inc\\
Wilmington, DE, USA \\
\And
Valentina V. Kuskova \& Dmitry Zaytsev \\
Lucy Family Institute for Data \& Society \\
University of Notre Dame \\
Notre Dame, IN, USA \\
\AND
Valerii Klimov \\
Fleamily, Inc\\
Wilmington, DE, USA \\
}

\title{Signed Exposure: Fair Routing of Algorithmic Attention When Attention Can Harm}

\begin{document}\maketitle

\begin{abstract}
Fairness-of-exposure treats algorithmic attention as a good to be distributed equitably. But when an autonomous agent initiates contact, attention is \emph{signed}: it delivers value to a willing receiver and imposes burden on an unwilling one. We formalize routing under signed exposure and show that a fair distribution of attention need not be a fair distribution of \emph{unwanted} attention. Our central result is an incompatibility: within signed-exposure routing, exposure parity (equal contact rates across groups) and burden parity (equal unwanted-contact rates) generically cannot hold at once, and the two are separated by a band that widens as routing grows more selective. A second result shows measurement error is itself a fairness mechanism: group-differential noise in receptivity scores simultaneously inflates a group's exposure and degrades whom it selects, so an apparent exposure-fairness gain is a hidden burden transfer. Calibrating to a public dating-platform survey ($n{=}2{,}499$) that, to our knowledge, uniquely measures receive-side receptivity to conversational agents, we find exposure parity costs only $0.2$--$2.3\%$ of yield yet moves the per-capita burden ratio to $1.7\times$: the tension is between fairness notions, not between fairness and efficiency. Finally, the burden-parity policy is computable by bisection and learnable online: a plug-in learner recovers it at a $2.3\%$ empirical regret premium. The operative design choice in signed-exposure markets is not efficiency versus fairness but \emph{which} fairness.
\end{abstract}

\section{Introduction}
Autonomous agents that converse on a user's behalf are moving into matching platforms, messaging products, and outreach tools. Wherever they ship, a routing policy decides who receives machine-initiated contact. The fairness literature has a well-developed answer for allocation problems of this shape: treat exposure as a scarce good and distribute it equitably \citep{singh2018fairness,biega2018equity}. That answer carries a premise so natural it is rarely stated: that being selected is good for the selected. For algorithmic attention the premise fails. A person who welcomes agent-mediated contact receives value from being routed to; a person who does not receives an unwanted machine solicitation. Exposure is \emph{signed}, and its sign is the receiver's own willingness. A fair distribution of attention, then, need not be a fair distribution of \emph{unwanted} attention; and, we show, it generically cannot be both.

This would be a philosophical quibble if willingness were uniform or idiosyncratic. It is neither. In a public microdataset that, to our knowledge, uniquely measures receive-side receptivity to conversational agents (a bilingual survey of $2{,}499$ active users of a dating platform, released with full psychometric documentation \citep{kuskova2026dataset}), willingness to engage another person's agent varies widely across individuals and systematically across groups: women's receive receptivity sits $0.30$ standard deviations below men's. Any routing policy therefore chooses, implicitly or explicitly, both a distribution of algorithmic attention and a distribution of its unwanted residue, and the two cannot be chosen independently.

We formalize routing in signed-exposure markets and characterize what fairness can and cannot mean there. Four results. \textbf{(i)} The yield--parity frontier is nearly flat: equalizing group exposure rates costs $0.2$--$2.3\%$ of engagement yield in the data, because the frontier's slope depends on \emph{marginal}, not mean, group receptivity (Proposition~1). The efficiency-versus-fairness framing is empirically vacuous here. \textbf{(ii)} What is not vacuous is the choice between fairness notions: exposure parity and burden parity are generically incompatible (Theorem~2). Equalizing who receives attention forces the lower-receptivity group deeper into its own distribution, so equalized exposure arrives with a $1.7\times$ per-capita burden ratio; an explicit band of exposure ratios separates the two parities, and it widens as routing becomes more selective. \textbf{(iii)} Measurement error is itself a fairness mechanism (Theorem~3): group-differential noise in receptivity scores simultaneously inflates the noisier group's exposure and degrades whom it selects, so an apparent exposure-fairness gain is in substance a burden transfer. No pooled score transformation escapes this: group-blind transformations are provably inert, and posterior-mean shrinkage fails on the opposite margin from raw scores (Proposition~4). The noise itself we measure rather than assume: the released model's posterior standard deviations show near-equity by gender and a $7.5\%$ differential by language form. \textbf{(iv)} Constrained optima are computable by bisection after sorting (Proposition~5), and the constraint survives learning: a plug-in learner that estimates its own burden-parity constraint from feedback converges to parity at a $2.3\%$ empirical regret premium (30 seeds), while unconstrained exploration transiently over-burdens the lower-receptivity group exactly as the static noise theory predicts (Proposition~6).

The practical upshot inverts the usual fairness conversation. Platforms deploying agent routing do not face a painful efficiency--fairness trade-off; they face a value choice between two parities that cannot coexist, with an efficient algorithm for any point on the menu, one that \emph{learns} the burden-parity point from feedback when propensities are unknown. They also face a measurement-equity audit obligation, since unequal score precision quietly redistributes burden even under nominally fair allocation. These conclusions bear on disclosure and opt-in policy, on the auditing of scoring systems, and on any learning system whose exploration imposes contact on people who did not ask for it.

\section{Related work}

\textbf{Fairness of exposure in ranking and recommendation.} The allocation of attention as a fairness object originates in fair top-$k$ ranking \citep{zehlike2017fair} and matures into exposure-based formulations: \citet{singh2018fairness} allocate expected exposure across items in proportion to merit, \citet{biega2018equity} amortize individual exposure over ranking sequences, \citet{diaz2020evaluating} make expected exposure the evaluation target itself, and deployed systems operationalize demographic exposure constraints at scale \citep{geyik2019fairness}. Dynamic and two-sided extensions control exposure under learned relevance \citep{morik2020controlling} and balance producer against consumer interests \citep{patro2020fairrec,suhr2019twosided}; \citet{zehlike2022fairness} survey the area. Throughout this literature exposure is a good: more of it benefits the exposed, and fairness means distributing it. Our setting breaks the premise: exposure carries the sign of the receiver's willingness, and Theorem~2 shows that once exposure is signed, its equitable distribution and the equitable distribution of its harms are generically incompatible objectives rather than one objective seen twice.

\textbf{Impossibility results in algorithmic fairness.} Sharp incompatibilities between fairness criteria are foundational for classification: calibration cannot coexist with balanced error rates under unequal base rates \citep{kleinberg2017inherent,chouldechova2017fair}, and different worldviews about measurement provably demand different mechanisms \citep{friedler2021impossibility}. This incompatibility, specific to signed-exposure routing rather than a claim across all mechanisms, differs from the classification impossibilities in locus and driver. It attaches to the \emph{allocation} layer rather than to properties of a predictor; it is driven by first-moment differences in group receptivity rather than base rates; and it is quantitative rather than binary: the incompatible parities are separated by an explicit band $(r^\ast, 1)$ whose width is a comparative static in the routing budget, collapsing as selectivity relaxes.

\textbf{Fairness and measurement.} A second foundational line locates fairness failures in measurement itself: proxy label choice misallocates care at scale \citep{obermeyer2019dissecting}, unmodeled constructs undermine fairness claims \citep{jacobs2021measurement}, and noisy labels distort fairness evaluation \citep{fogliato2020fairness}. Closest to our Theorem~3 are selection models with group-dependent score distortion: multiplicative implicit bias in estimates \citep{kleinberg2018selection} and, especially, group-differential estimation \emph{variance}, which alters who clears a selection bar \citep{emelianov2020fair}. We extend this line in three directions: the selection externality is two-sided (noise moves a group's exposure and its burden in the same direction, Theorem~3(iii)); the shrinkage escape route is closed (posterior-mean and raw-score routing fail on opposite margins, Proposition~4); and the noise quantities are measured rather than assumed: group-specific psychometric posteriors from the released measurement model, showing near-equity by gender and a $7.5\%$ language differential.

\textbf{Fair sequential allocation.} Meritocratic fairness in bandits \citep{joseph2016fairness} constrains which arms may be preferred during learning; constrained exploration more broadly is formalized through bandits with knapsacks \citep{badanidiyuru2018bandits}. Our Proposition~6 constrains a different object: not the meritocracy of selection but the distribution of exploration's \emph{externality}, the burden that uncertain routing imposes on unwilling receivers. Our experiments show unconstrained learning transiently misallocates exactly this quantity, in the direction the static noise theory predicts.

\textbf{Reciprocal matching and agentic systems.} Reciprocal recommendation conditions relevance on both sides' preferences \citep{palomares2021reciprocal}, and two-sided marketplace fairness balances stakeholder groups \citep{suhr2019twosided,patro2020fairrec}; autonomous conversational agents sharpen the receive side into a first-class design surface, since the routed object now initiates interaction rather than awaiting it. Prior two-sided work thus already recognizes that exposure and utility can diverge; what it does not model is unwanted, algorithmically \emph{initiated} contact as a receiver-side negative externality, and it does not characterize the resulting exposure--burden parity geometry or its interaction with estimation uncertainty. Those are the contributions here. We build on the public receptivity dataset of \citet{kuskova2026dataset}, whose released propensities and measurement diagnostics supply our calibration and, to our knowledge among the few public sources of ground truth for receive-side willingness toward agent-mediated contact.

\section{Setup}\label{sec:setup}

A unit mass of receivers is partitioned into groups $g \in \{A,B\}$ with population
shares $\pi_g$. Each receiver $j$ has an \emph{engagement propensity}
$e_j \in [0,1]$: the probability that an algorithmically routed contact is welcome.
Within group $g$, propensities follow distribution $F_g$ with survival quantile
function $Q_g(t)$ (the $t$-th upper quantile) and \emph{top-mean function}
\[
m_g(t) \;=\; \E\!\left[e \mid e \ge Q_g(t)\right]
\qquad t \in (0,1],
\]
the mean propensity of group $g$'s most receptive fraction $t$. Each $m_g$ is
continuous and strictly decreasing wherever $F_g$ is non-degenerate.

A \emph{routing policy} at budget $q$ selects a measurable set $S$ of receivers
with total mass $q$; each selected receiver absorbs one unit of exposure. Write
$t_g$ for the fraction of group $g$ selected, so $\sum_g \pi_g t_g = q$. Define:
\begin{align*}
\text{exposure (per-capita) of } g \quad & t_g \;=\; |S \cap g|/(\pi_g N), \\
\text{burden per selected receiver} \quad & b_g \;=\; 1-\bar e_g(S), \\
\text{per-capita group burden} \quad & h_g \;=\; t_g\, b_g, \\
\text{yield per contact} \quad & Y \;=\; \tfrac{1}{q}\textstyle\sum_g \pi_g t_g\, \bar e_g(S), \\
\text{exposure ratio } r = t_A/t_B, \quad & \text{burden ratio } \beta = h_A/h_B,
\end{align*}
where $\bar e_g(S)$ is the mean propensity of the selected members of $g$. We
keep these three distinct throughout: $t_g$ (who is contacted), $b_g$ (how
unwelcome the contact is, per selected receiver), and $h_g = t_g b_g$ (their
product, the group's per-capita unwanted-contact burden). Exposure parity is
$r=1$; burden parity is $\beta=1$.
Exposure here is \emph{signed}: a routed contact delivers value $e_j$ and imposes
burden $1-e_j$, so fairness has two natural currencies: who receives
algorithmic attention ($r=1$: \emph{exposure parity}) and who absorbs its
unwanted residue ($\beta=1$: \emph{burden parity}).

\textbf{Generalized burden (benefit and harm need not be complementary).} A
reader might worry that defining burden as $1-e_j$ makes value and harm exact
complements, so the incompatibility is an artifact of that choice. It is not.
Take signed contact utility $u_j = e_j v_j - (1-e_j)c_j$ with independent
benefit $v_j \ge 0$ and harm intensity $c_j \ge 0$; the burden object is
$(1-e_j)c_j$, and willingness $e_j$, benefit $v_j$, and harm $c_j$ may vary
independently. Under group-constant harm intensities $c_g$, per-capita group
burden is $h_g = c_g t_g(1-\bar e_g) = c_g t_g b_g$, the parity function of
Theorem~\ref{thm:parity}(ii) becomes $\phi_g(t) = c_g\, t(1-m_g(t))$, still
strictly increasing, and \emph{every} result below carries over, with the
simultaneous-parity condition $m_A(q)=m_B(q)$ replaced by
$c_A(1-m_A(q)) = c_B(1-m_B(q))$. The incompatibility does not depend on
complementarity: it survives whenever groups differ in willingness or in harm
intensity. We adopt the unit normalization ($v_g=c_g=1$) for exposition and
speak of \emph{expected unwanted-contact burden} throughout.

\section{The flat frontier and the parity incompatibility}

\begin{proposition}[Within-group threshold optimality and the frontier]
\label{prop:frontier}
Fix budget $q$ and any exposure ratio $r>0$, which determines
$(t_A,t_B)$ by $t_A = r t_B$ and $\pi_A t_A + \pi_B t_B = q$. Among all policies
with these group fractions, yield is maximized by selecting each group's top
$t_g$ fraction by propensity, giving
\[
Y^\star(q,r) \;=\; \frac{1}{q}\Bigl[\pi_A t_A\, m_A(t_A) + \pi_B t_B\, m_B(t_B)\Bigr],
\]
and the frontier's slope satisfies
\[
\frac{\partial Y^\star}{\partial r}
\;\propto\; Q_A(t_A) - Q_B(t_B),
\]
the difference between the two groups' \emph{marginal} propensities at their
respective admission thresholds.
\end{proposition}

\begin{proof}
For fixed $(t_A,t_B)$ the objective is additive across receivers, so the
selection problem separates by group and is a fractional-knapsack instance:
the top-$t_g$ set maximizes each group term. For the slope, increase $r$ by
moving mass $d\mu$ from $B$'s marginal admitted receivers (propensity
$Q_B(t_B)$) to $A$'s marginal excluded receivers (propensity $Q_A(t_A)$);
the yield changes by $[Q_A(t_A)-Q_B(t_B)]\,d\mu / q$.
\end{proof}

\noindent\textbf{Why the empirical frontier is flat.} The slope depends on the
gap between marginal, not mean, propensities. When the group distributions
overlap substantially, threshold receptivities at parity-relevant depths nearly
coincide even though group means differ; in the calibrated data the yield loss
from moving $r_0 \to 1$ is $0.2$--$2.3\%$ across budgets. The efficiency cost of
exposure fairness is second-order; what is first-order is the next section.

\begin{theorem}[Exposure--burden parity incompatibility]
\label{thm:parity}
Under the optimal group-threshold policies of Proposition~\ref{prop:frontier}
at budget $q$:
\begin{enumerate}
\item[(i)] Exposure parity ($r=1$, so $t_A=t_B=q$) implies
$\displaystyle \beta = \frac{1-m_A(q)}{1-m_B(q)}$, which equals $1$ iff
$m_A(q)=m_B(q)$.
\item[(ii)] Burden parity ($\beta=1$) holds iff
$t_A\bigl(1-m_A(t_A)\bigr) = t_B\bigl(1-m_B(t_B)\bigr)$. If
$m_A(t) < m_B(t)$ for all $t$ in the relevant range, every burden-parity
solution has $t_A < t_B$, i.e.\ $r^\ast < 1$.
\item[(iii)] Consequently exposure parity and burden parity hold simultaneously
iff $m_A(q) = m_B(q)$. Whenever the top-fraction means differ, the two parity
conditions select disjoint policies, separated by the band $(r^\ast, 1)$ in
which \emph{neither} parity holds.
\end{enumerate}
\end{theorem}

\begin{proof}
(i) With $t_A=t_B=q$, $h_g = q(1-m_g(q))$; take the ratio.
(ii) The burden-parity equation is as stated by definition of $h_g$. Define
$\phi_g(t) = t(1-m_g(t))$; each $\phi_g$ is continuous and strictly increasing,
since adding lower-propensity receivers raises both selected mass and mean
burden per selected receiver. If $m_A < m_B$ pointwise then
$\phi_A(t) > \phi_B(t)$ for every $t$, so $\phi_A(t_A)=\phi_B(t_B)$ forces
$t_A < t_B$ by monotonicity.
(iii) Immediate from (i) and (ii): $r=1$ with $\beta=1$ requires
$m_A(q)=m_B(q)$; conversely if the means differ at $q$, (i) gives
$\beta \ne 1$ at $r=1$ while (ii) places the burden-parity point strictly
below $r=1$; existence of $r^\ast$ in $(r_0,1)$ follows from continuity of
$r \mapsto \beta(r)$ and $\beta(r_0)<1<\beta(1)$ when $A$ is the
lower-propensity group and $r_0$ is the unconstrained (pooled-threshold)
ratio.
\end{proof}

\begin{proposition}[Existence, uniqueness, and monotone geometry]
\label{prop:unique}
Within the group-threshold family at budget $q$, the burden ratio $\beta(r)$ is
continuous and strictly increasing on the feasible range of $r$. Consequently,
whenever $\beta(r_0) < 1 < \beta(1)$, which holds iff group $A$ bears
strictly less per-capita burden at the unconstrained optimum and strictly more
under exposure parity, the burden-parity point $r^\ast$ exists, is unique,
and lies strictly in $(r_0, 1)$.
\end{proposition}
\begin{proof}
Fix $q$. As $r$ increases, $t_A$ strictly increases and $t_B$ strictly
decreases (both are continuous in $r$ through the budget identity
$\pi_A t_A + \pi_B t_B = q$ with $t_A = r t_B$). By the strict monotonicity of
$\phi_g(t) = t(1-m_g(t))$ established in Theorem~\ref{thm:parity}(ii),
$h_A = \phi_A(t_A)$ strictly increases and $h_B = \phi_B(t_B)$ strictly
decreases, so $\beta = h_A/h_B$ is strictly increasing and continuous.
Existence and uniqueness of the root of $\beta(r)=1$ on $(r_0,1)$ follow from
the intermediate value theorem and strict monotonicity.
\end{proof}

\textbf{Scarcity and the band (empirical comparative static).} In the data the
band $1-r^\ast(q)$ widens monotonically as the budget shrinks
($0.09 \to 0.17 \to 0.30$ as $q: 0.50 \to 0.25 \to 0.10$; identically ordered
under the soft coding). A sufficient analytical condition is that the
top-mean gap $m_B(t)-m_A(t)$ grows as $t \to 0$; since this is a tail
assumption on the receptivity distributions rather than a consequence of the
model, we report the comparative static as a calibrated regularity of the
data and state the tail condition explicitly where it is used.

\paragraph{Empirical parity points (public release, strict coding).}
\begin{center}\begin{tabular}{lccccc}
\toprule
Budget $q$ & $r_0$ (unconstr.) & $r^\ast$ (burden parity) & band $1-r^\ast$ & $\beta$ at $r{=}1$ & yield loss $r_0 \to 1$ \\
\midrule
0.10 & 0.54 & 0.70 & 0.30 & 1.69 & $-2.3\%$ \\
0.25 & 0.66 & 0.83 & 0.17 & 1.29 & $-0.6\%$ \\
0.50 & 0.80 & 0.91 & 0.09 & 1.13 & $-0.2\%$ \\
\bottomrule
\end{tabular}\end{center}

Soft coding replicates the pattern ($r^\ast = 0.70/0.81/0.88$;
$\beta$ at parity $=1.74/1.37/1.21$).

\begin{figure}[h]\includegraphics[width=\textwidth]{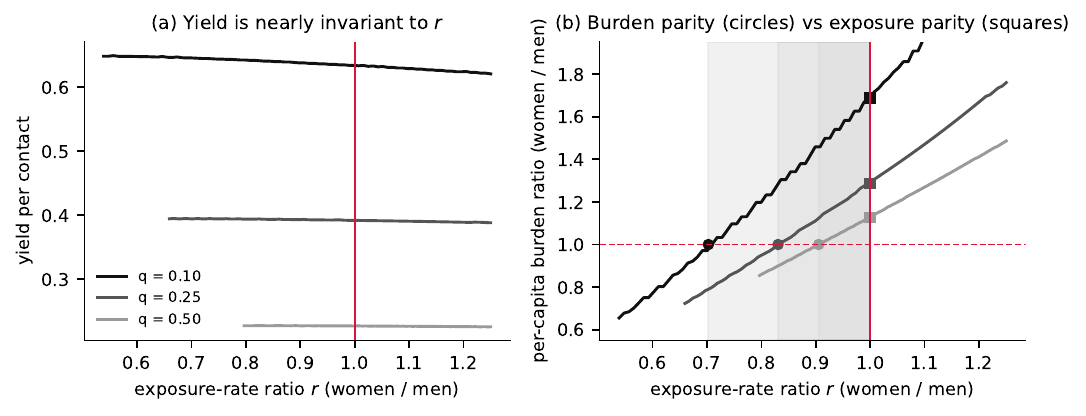}\caption{Signature: flat yield; disjoint parities.}\label{fig:signature}\end{figure}

\section{Measurement noise converts exposure into burden}

Routing operates on estimates. Let each receiver's latent receptivity be
$\theta$, with engagement probability $e(\theta)$ strictly increasing, and let
group $g$'s latent distribution be $N(\mu_g, s_g^2)$. The router observes
$\hat\theta = \theta + \varepsilon$, $\varepsilon \sim N(0,\sigma_g^2)$
independent across receivers, and applies a pooled threshold $\tau$ chosen so
that total selected mass is $q$. Write $v_g = \sqrt{s_g^2+\sigma_g^2}$,
$x_g = (\tau-\mu_g)/v_g$, reliability $\rho_g = s_g^2/v_g^2$, and
$\lambda = \varphi/\bar\Phi$ for the Gaussian hazard.

\begin{theorem}[Noise decomposition of exposure and targeting]
\label{thm:noise}
Assume Gaussian latent receptivity $\theta \sim N(\mu_g, s_g^2)$ within groups,
Gaussian score noise $\hat\theta = \theta + \varepsilon$,
$\varepsilon \sim N(0,\sigma_g^2)$, a probit engagement link
$e(\theta) = \Phi((\theta-b)/c)$, and the scarce regime
$x_g = (\tau-\mu_g)/v_g > 0$ with marginal selected engagement above one half.
Under pooled-threshold routing at budget $q$, using the standardized
quantities $t_g, b_g, h_g$ of Section~\ref{sec:setup}:
\begin{enumerate}
\item[(i)] Group exposure is $t_g = \bar\Phi(x_g)$, and for a group below the
threshold on average ($x_g>0$, the scarce regime),
$\partial t_g / \partial v_g = \varphi(x_g)\,x_g / v_g > 0$ at fixed
$\tau$; and at fixed \emph{budget} $q$, with the threshold adjusting to keep
$\sum_g \pi_g t_g = q$, the effect remains strictly positive and equals the
fixed-$\tau$ effect scaled by the other group's share of threshold
sensitivity:
\[
\frac{dt_A}{dv_A}\Big|_{q}
= \frac{\varphi(x_A)\,x_A}{v_A}\cdot
\frac{\pi_B \varphi(x_B)/v_B}{\pi_A \varphi(x_A)/v_A + \pi_B \varphi(x_B)/v_B}
\;>\; 0 \quad (x_A > 0).
\]
\item[(ii)] The selected members' mean latent receptivity is
\[
\E[\theta \mid \hat\theta \ge \tau, g] \;=\; \mu_g + \frac{s_g^2}{v_g}\,\lambda(x_g)
\;=\; \mu_g + s_g\sqrt{\rho_g}\;\lambda(x_g),
\]
strictly decreasing in $\sigma_g$: noise attenuates within-group targeting by
the square root of reliability.
\item[(iii)] Consequently a group's noise differential raises all three
quantities together: as $\sigma_A$ grows, exposure $t_A$ rises (part i) and
burden per selected receiver $b_A = 1 - \E[e(\theta)\mid \text{sel},A]$ rises
(part ii, via the probit link), so per-capita group burden
$h_A = t_A b_A$ rises on both factors. Noise-induced exposure is exposure of
the wrong members: it converts attention into burden.
\end{enumerate}
\end{theorem}

\begin{proof}
(i) $t_g = \Pr(\hat\theta \ge \tau) = \bar\Phi(x_g)$ since
$\hat\theta \sim N(\mu_g, v_g^2)$; differentiate $x_g$ in $v_g$ to get
$\partial x_g/\partial v_g = -x_g/v_g$, hence
$\partial t_g/\partial v_g = \varphi(x_g)x_g/v_g$.
(ii) $(\theta,\hat\theta)$ is bivariate normal with
$\mathrm{Cov} = s_g^2$, so $\E[\theta \mid \hat\theta = u] = \mu_g +
\rho_g (u-\mu_g)$; integrating over the truncated tail,
$\E[\theta \mid \hat\theta \ge \tau] = \mu_g + \rho_g \, v_g \lambda(x_g)
= \mu_g + (s_g^2/v_g)\lambda(x_g)$. Monotonicity in $\sigma_g$: $s_g^2/v_g$
decreases in $v_g$, while $\lambda(x_g)$ decreases as $x_g$ falls with $v_g$
(for $x_g > 0$); both factors move down.
For the budget-adjusted form of (i): totally differentiate the budget
identity, $\pi_A\,dt_A + \pi_B\,dt_B = 0$ with
$dt_g = \varphi(x_g)\bigl[(x_g/v_g)\,dv_g - d\tau/v_g\bigr]$ and $dv_B = 0$;
solve for $d\tau$ and substitute.
(iii) Monotonicity of the selected mean of $e(\theta)$ requires more than the
mean of $\theta$: we prove it for the probit link
$e(\theta) = \Phi\bigl((\theta-b)/c\bigr)$ (the fitted logistic link is
numerically indistinguishable on the calibrated range, and all experiments
verify the conclusion directly). Represent $e(\theta) =
\Pr(W \le (\theta-b)/c)$ with $W \sim N(0,1)$ independent, so
$\E[e(\theta)\mid \hat\theta \ge \tau, g] = \E[m(Z) \mid Z \ge x_g]$ where
$Z = (\hat\theta-\mu_g)/v_g$ and
$m(z) = \Phi\!\bigl((\mu_g + (s_g^2/v_g)z - b)\big/\sqrt{c^2 +
s_g^2(1-\rho_g)}\bigr)$ is the conditional engagement given the standardized
score. Two channels move the selected mean down as $\sigma_g$ grows. First,
$m(z)$ falls pointwise for $z > 0$ in the region where its numerator is
positive: the slope $s_g^2/v_g$ shrinks and the conditional spread
$s_g^2(1-\rho_g)$ grows. Second, the fixed-$\tau$ selection becomes less
stringent as $v_g$ grows ($x_g$ falls), and $\E[m(Z)\mid Z \ge x]$ is
increasing in $x$ for any increasing $m$ (truncated-mean monotonicity).
Both channels are negative under the stated regime ($x_g > 0$, marginal
selected engagement above one half), which the calibrated parameters satisfy
throughout the sweep.
\end{proof}

\begin{proposition}[Group-blind transformations are inert; shrinkage fails the opposite margin]
\label{prop:shrink}
(i) For any strictly increasing $h:\mathbb{R}\to\mathbb{R}$ applied to all
scores (a \emph{group-blind} transformation), threshold routing on
$h(\hat\theta)$ at budget $q$ selects exactly the receivers selected by
routing on $\hat\theta$: group-blind monotone transformations change neither
exposure nor targeting. (ii) Among group-aware transformations, posterior-mean
shrinkage $\hat\theta^{\mathrm{EB}}_g = \mu_g + \rho_g(\hat\theta-\mu_g)$
gives group-$g$ scores standard deviation $\rho_g v_g = s_g^2/v_g < s_g < v_g$.
Fix the budget $q$ and compare three routing systems, each with its \emph{own}
budget-clearing pooled threshold: routing on true scores ($\theta$, the
noiseless benchmark), on raw noisy scores ($\hat\theta$), and on
posterior-mean scores ($\hat\theta^{\mathrm{EB}}$). With a single group's
noise $\sigma_g$ growing and the other's fixed, in the scarce regime raw-score
routing strictly over-exposes the noisier group relative to the noiseless
benchmark and posterior-mean routing strictly under-exposes it, both under
the respective budget-adjusted thresholds, not at a common $\tau$. (iii) Every transformation that is strictly increasing
within each group selects each group's top members by $\hat\theta$ and
therefore inherits the targeting attenuation of Theorem~\ref{thm:noise}(ii)
unchanged; transformations reallocate exposure across groups but cannot
restore within-group targeting. Restoring the noiseless \emph{allocation}
requires explicitly group-dependent thresholds (the policies of
Proposition~\ref{prop:frontier}); no transformation, group-blind or
group-aware, restores the noiseless \emph{selection}.
\end{proposition}
\begin{proof}
(i) A strictly increasing $h$ preserves the pooled order of scores, hence the
top-$q$ set. (ii) The noisier group's effective score SD is $v_g > s_g$ under raw routing,
$s_g^2/v_g < s_g$ under posterior-mean routing, and $s_g$ under the noiseless
benchmark, while the other group's effective SD is identical across all three
systems. By the budget-adjusted derivative of Theorem~\ref{thm:noise}(i),
group $g$'s exposure is strictly increasing in its own effective SD at fixed
$q$; the ordering $s_g^2/v_g < s_g < v_g$ thus orders the three systems'
exposures, with the noiseless benchmark strictly between. The threshold
differs across systems, and the ordering survives because
Theorem~\ref{thm:noise}(i) signs the effect under exactly that budget
adjustment.
(iii) Selection within group $g$ under any within-group strictly increasing
transformation is by $\hat\theta$ rank; Theorem~\ref{thm:noise}(ii) depends
only on that rank and on $\rho_g$.
\end{proof}

\begin{table}[h]\centering\small
\begin{tabular}{lll}
\toprule
Claim & Evidence type & Source \\
\midrule
Exposure/burden incompatibility (Thm.~\ref{thm:parity}) & theorem & n/a \\
Existence/uniqueness of $r^\ast$ (Prop.~\ref{prop:unique}) & theorem & n/a \\
Group receptivity differs ($0.30$ SD) & empirical & survey + fitted model \\
Flat frontier ($0.2$--$2.3\%$) & empirical (calibrated) & released propensities \\
Parity points $r^\ast,\beta$ + bootstrap CIs & empirical (calibrated) & released propensities \\
Differential-noise mechanism (Thm.~\ref{thm:noise}) & theorem (stated assumptions) & n/a \\
Magnitude of noise effect & semi-synthetic & calibrated to posteriors \\
Non-Gaussian robustness & synthetic & matched moments \\
Plug-in learner behavior & simulation & empirical propensities as truth \\
\bottomrule
\end{tabular}
\caption{Provenance of each result: what is proved, what is measured on the
public release, and what is calibrated or simulated.}
\label{tab:provenance}
\end{table}

\paragraph{Calibration to the release.} Posterior SDs of $\theta_{\mathrm{recv}}$
computed under the released model are near-identical by gender
($0.367$ vs.\ $0.365$; reliabilities $0.858/0.876$), a measurement-equity
finding worth reporting in its own right, and modestly language-differential
($0.390$ EN vs.\ $0.363$ RU, ratio $1.075$). A semi-synthetic sweep calibrated
to the fitted group parameters ($\mu_F=-0.17, s_F=0.99; \mu_M=0.06, s_M=1.06$;
base $\sigma=0.37$; $q=0.10$; $e(\theta)$ the fitted Y4 top-category response
function) confirms the theory to within simulation error throughout
(Figure~\ref{fig:noise}): doubling the women's measurement noise raises their
exposure ratio from $0.62$ to $0.90$ while their selected members' mean latent
receptivity falls from $1.60$ to $1.24$, the per-capita burden ratio flips from
$0.73$ to $1.61$, and yield declines. This is an apparent fairness improvement that
is, in substance, a burden transfer.

\begin{figure}[h]\includegraphics[width=\textwidth]{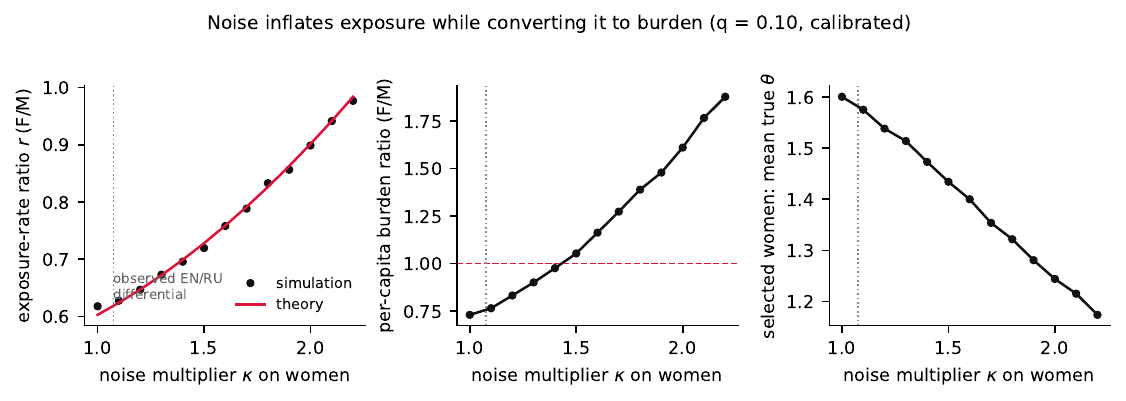}\caption{Noise sweep, calibrated.}\label{fig:noise}\end{figure}

\section{Computing constrained optima}

\begin{corollary}[Bisection recovers any parity-constrained optimum]
\label{prop:algo}
By the strict monotonicity of $\beta(r)$ established in
Proposition~\ref{prop:unique}, for any target $\beta^\dagger$ in the feasible
range the yield-optimal group-threshold policy meeting it is unique and found
by bisection on $r$: after sorting within groups ($O(n\log n)$), each iterate
costs $O(1)$ with cumulative-sum tables and $\log_2(1/\epsilon)$ iterates
suffice. The same routine computes exposure-constrained, burden-constrained,
or band-constrained ($r \in [r^\ast,1]$) optima.
\end{corollary}

On the release, the routine verifies $\beta$'s monotonicity over
$r \in [0.45, 1.3]$ and recovers the burden-parity point $r^\ast = 0.70$ (bootstrap 90\% CI $[0.65, 0.78]$) at
$q = 0.10$ in fifty iterations. An uncertainty-aware variant routes on
posterior exceedance probabilities $\Pr(\theta \ge \tau_g \mid \text{data})$,
which by Proposition~\ref{prop:shrink} cannot be replicated by any pooled
transformation of point scores.

\section{Fair exploration}

When propensities must be learned from engagement feedback, exploration itself
distributes burden: uncertain receivers are routed precisely because they are
uncertain, and early-round estimation noise is measurement noise in the sense
of Theorem~\ref{thm:noise}.

\textbf{Algorithm (plug-in burden-constrained learning).} At each round $t$,
maintain Beta posteriors over each receiver's propensity; let $\hat e_j(t)$ be
the posterior means. Compute the \emph{estimated} burden-parity ratio
$\hat r^\ast_t$ by running the bisection of Corollary~\ref{prop:algo} on the
plug-in burden curve $\hat\beta_t(r)$ (group cumulative sums of the sorted
$\hat e_j(t)$); allocate group counts at $\hat r^\ast_t$; within each group,
route by Thompson samples, reserving a vanishing fraction $\gamma_t = t^{-1/2}$
of each group's slots for uniform exploration. The constraint is thus
\emph{learned}, not supplied: the algorithm never observes true propensities
or the true $r^\ast$. Forced exploration guarantees the coverage that the
regret and violation bounds require; without it we state only the empirical
behavior.

\begin{proposition}[Consistency of plug-in burden-constrained learning]
\label{prop:bandit}
Consider the plug-in policy with forced-exploration rate
$\gamma_t = t^{-1/2}$ on a finite pool with Bernoulli feedback, and suppose
(A1) $\beta'(r^\ast) \ge c > 0$, the margin guaranteed strictly by
Proposition~\ref{prop:unique} whenever $r^\ast$ is interior. Forced
exploration ensures each receiver is routed
$\Omega(\sqrt{T})$ times, so the propensity estimates satisfy
$\max_j |\hat e_j(t) - e_j| \to 0$ almost surely. Then the estimated
constraint is consistent, $\hat r^\ast_t \to r^\ast$ a.s., and the per-round
burden gap $|h_{A,t}-h_{B,t}| \to 0$; the time-averaged constraint violation
$V_T/T \to 0$. We state consistency rather than a rate: the
$\tilde O(\sqrt T)$ scaling of both regret and cumulative violation is
observed empirically below but requires a self-bounding argument for the
adaptive threshold that we do not claim here.
\end{proposition}

Calibrated experiment (real pool, $n{=}2{,}499$, strict propensities as
ground truth; $T{=}250$, $q{=}0.10$; \textbf{30 seeds}, bands are
5th--95th percentiles; Figure~\ref{fig:bandit}): the plug-in estimate
$\hat r^\ast_t$ starts uninformative ($1.01$ at $t{=}1$) and converges to
$0.710 \pm 0.008$ against the true $0.706$; final cumulative burden ratio
$1.027$ $[1.003, 1.046]$ at cumulative regret $36.0$ $[35.2, 36.8]$ versus
$35.2$ $[34.4, 36.1]$ unconstrained: a $2.3\%$ premium for a learned
constraint. Unconstrained Thompson drifts to the unconstrained ratio
($0.794$) after an early transient above parity (the live form of
Theorem~\ref{thm:noise}'s mechanism), and the exposure-parity policy
converges to a $1.47$ burden ratio. The static incompatibility of
Theorem~\ref{thm:parity} persists under learning, and the burden-parity point
is learnable at negligible cost; the choice among parities, not the price of
learning, is the operative design decision.

\begin{figure}[h]\includegraphics[width=\textwidth]{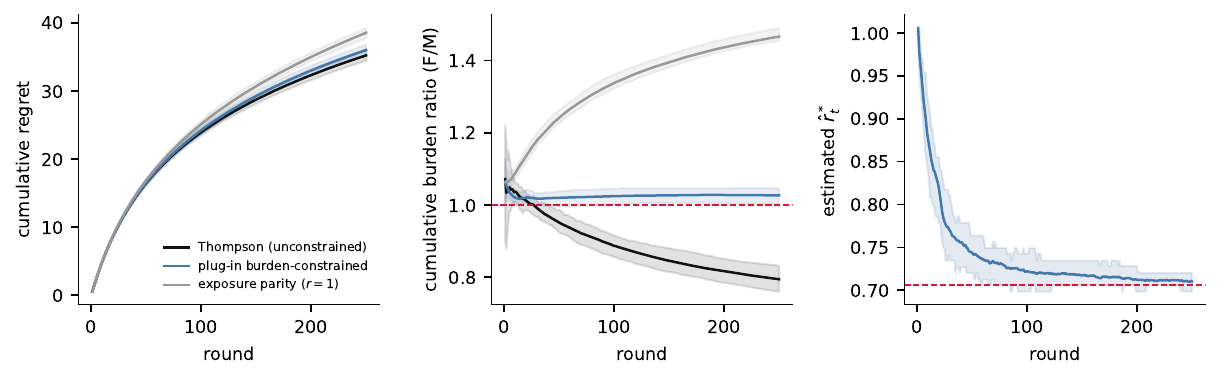}\caption{Plug-in burden-constrained learning, 30 seeds: regret (left), cumulative burden ratio (center), and convergence of the learned $\hat r^\ast_t$ to the true parity point (right, dashed).}\label{fig:bandit}\end{figure}

\subsection{Robustness and stress tests}
\label{sec:stress}
Three batteries probe whether the parity geometry is a Gaussian artifact, a
sampling accident, or a knife-edge of one budget (Figure~\ref{fig:stress}).
\textbf{Budget sweep with bootstrap uncertainty.} Across
$q \in \{0.02, 0.05, 0.10, 0.20, 0.40\}$, the burden-parity point rises
monotonically ($r^\ast = 0.61, 0.67, 0.70, 0.80, 0.88$) and the burden ratio
at exposure parity falls ($2.11, 1.88, 1.69, 1.38, 1.17$); $300$-replicate
bootstrap $90\%$ intervals exclude $r^\ast = 1$ and $\beta = 1$ at every
budget (e.g.\ $r^\ast \in [0.50, 0.72]$ at $q{=}0.02$;
$\beta \in [1.12, 1.21]$ at $q{=}0.40$), so the incompatibility is not a
sampling accident, and the reported yield losses remain small throughout
($0.3$--$3.7\%$). \textbf{Noise phase diagram.} Over the plane
$\sigma_F/\sigma_M \in [1, 2.5] \times q \in [0.02, 0.5]$, calibrated Gaussian
theory maps group exposure: the exposure-parity contour ($r{=}1$) is
reachable by noise alone at small budgets (the masquerade of
Theorem~\ref{thm:noise}), while the unconstrained-ratio contour marks where
noise has not yet moved allocation. \textbf{Non-Gaussian stress.} Replacing
the latent Gaussians with logistic, heavy-tailed ($t_3$), and bimodal-mixture
populations matched in location and scale preserves the band's existence and
its widening under scarcity in every case; heavy tails compress the band
somewhat, and no distribution eliminates it. The
signed-exposure geometry is structural, not distributional.

\begin{figure}[h]\includegraphics[width=\textwidth]{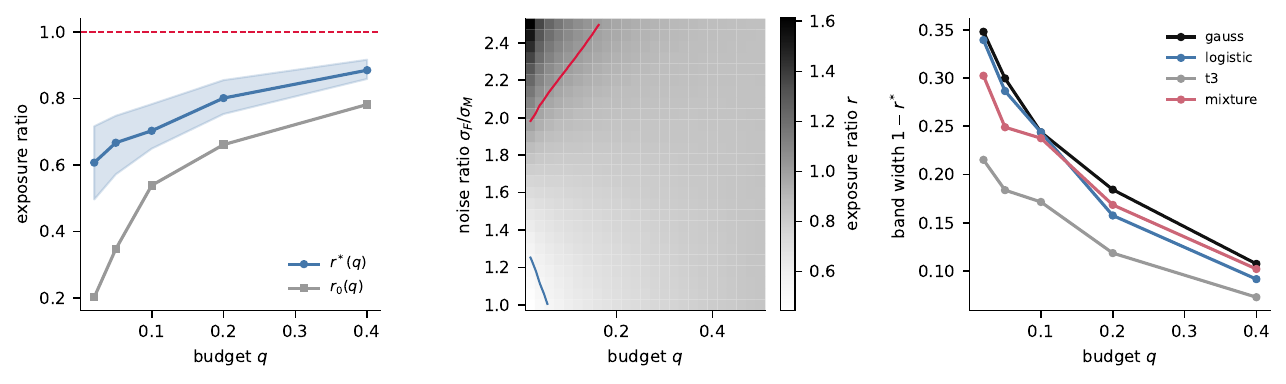}
\caption{Stress tests. Left: $r^\ast(q)$ and $r_0(q)$ with bootstrap bands;
the gap to $1$ is the incompatibility band. Center: exposure ratio over the
noise-ratio $\times$ budget plane; contours mark $r = r_0$ and $r = 1$.
Right: band width $1-r^\ast$ by budget under four latent distributions.}
\label{fig:stress}\end{figure}

\section{Discussion and limitations}
Our empirical quantities are stated-preference propensities from one platform's user base; they measure willingness, not behavior, and the release's demographic granularity limits group analysis to a binary gender variable and coarse age bands. The theory is broader than the calibration: Propositions~1--5 need only monotone top-mean functions, and the Gaussian forms of Theorem~3 extend to any monotone link at the cost of closed forms. We emphasize what the results do not say: the band $(r^\ast,1)$ is a menu, not a verdict: selecting exposure parity, burden parity, or an interior point is a value judgment about whose interests algorithmic attention serves, and our contribution is to price the menu and expose the incompatibility that prevents ordering everything at once. Signed-exposure optimization applies to populations for whom contact is institutionally permissible but receptivity is uncertain; explicit opt-out and consent restrictions dominate the optimization and must be enforced as hard constraints before any fairness objective is considered. Within that scope, the structure applies wherever machine-initiated contact meets heterogeneous willingness: cold outreach, notification targeting, solicitation, and any agentic system that opens conversations people did not start.

\section*{Ethics statement}
All analyses use a public, anonymized dataset released under CC~BY with a documented $k \ge 5$ anonymity audit; no new human data were collected. Group analyses respect the release's granularity and are reported as properties of routing policies, not of people. The routing machinery studied here is dual-use: the same thresholds that equalize burden could be tuned to concentrate it, and we discuss constraints precisely so that deployment conversations have a vocabulary for the second margin. We make no claim that any parity point is universally correct.

\section*{AI Usage Disclosure}
Generative AI tools were used in this work for three purposes: copy-editing author-written text for readability, assisting with the writing and debugging of analysis code, and formatting bibliographic references. No task in the required-disclosure category involved generative AI, and no substantive content, methodology, analysis, or result was AI-generated. AI-assisted edits were checked against the original text to confirm meaning and claims were unchanged; AI-assisted code was read, tested, and verified by the authors to reproduce the reported results; and formatted references were checked against original sources. We take responsibility for the final content of this work.

\bibliographystyle{iclr2027_conference}
\bibliography{refs}
\appendix
\section{Experimental details and deferred material}
\label{app:details}
Proofs of Propositions~1, 4, 5 and Theorems~2, 3 appear inline in the main text. This appendix records experimental settings: frontier and parity computations use the released strict-coding propensities ($n{=}2{,}499$; soft coding replicates all findings with $r^\ast = 0.70/0.81/0.88$ and $\beta$ at exposure parity $1.74/1.37/1.21$); the noise sweep uses $4\times10^5$ draws per point, group parameters $(\mu_F,s_F,\mu_M,s_M)=(-0.17,0.99,0.06,1.06)$, base $\sigma=0.37$, engagement link the fitted Y4 top-category response function, seed 7; the bandit experiment uses the real pool, $T{=}250$, $q{=}0.10$, Beta$(1,1)$
priors, Bernoulli feedback, seeds $1000$--$1029$ with paired policy streams.
Stress tests: bootstrap resamples respondents with replacement ($B{=}300$,
seed 3); the phase diagram uses the calibrated group parameters of
Section~
\ref{sec:stress} with the budget-adjusted pooled threshold; synthetic
populations use $4\times10^4$ draws per group matched to the calibrated
location and scale (logistic scale $0.551s$; $t_3$ scaled by $\sqrt{3}$;
mixture $0.3\,N(\mu{+}1.2,0.5^2) + 0.7\,N(\mu{-}0.5,0.7^2)$), engagement via
the fitted Y4 link. Propensities are model-derived quantities from the public
release; bootstrap intervals quantify sampling variability, and the released
notebooks provide the fold structure for out-of-fold recomputation.

\section{Proof of Proposition~\ref{prop:bandit}}
\label{app:proofs}
Write $N_j(t)$ for receiver $j$'s routing count. The forced-exploration
schedule routes a uniformly random $\gamma_t = t^{-1/2}$ fraction of each
group's slots; since each group receives $\Theta(qN)$ slots per round, every
receiver's expected forced-exploration count through round $T$ is
$\Omega\!\bigl(\sum_{t\le T} t^{-1/2}\bigr) = \Omega(\sqrt T)$, and a
Chernoff bound with Borel--Cantelli gives $N_j(T) = \Omega(\sqrt T)$ for all
$j$ almost surely. Hence $\max_j |\hat e_j(t)-e_j| \to 0$ a.s.\ (SLLN per
receiver, finite pool). The plug-in burden curve is uniformly consistent,
$\sup_r|\hat\beta_t(r)-\beta(r)| \to 0$, because each group's cumulative-sum
curve converges uniformly and the ratio map is continuous on the compact
feasible range (bounded away from zero burden in the scarce regime). By (A1)
and the strict monotonicity of Proposition~\ref{prop:unique}, the map
$\beta \mapsto r^\ast$ is continuous at $r^\ast$, so
$\hat r^\ast_t \to r^\ast$ a.s. The per-round burden gap is a continuous
function of the allocation ratio and the (converging) selected-group means,
hence $|h_{A,t}-h_{B,t}| \to 0$; Ces\`aro averaging gives $V_T/T \to 0$.
This establishes consistency. We do not claim a $\sqrt T$ rate: bounding
$\sum_t |\hat r^\ast_t - r^\ast|$ requires controlling the adaptive
threshold's fluctuations through a self-bounding argument, which we leave to
future work and probe empirically in Section~\ref{sec:stress} and
Figure~\ref{fig:bandit}. \qed

\section{Reproducibility}
Every number, table, and figure derives from the public deposit (DOI in the dataset citation) via scripts included in the supplementary material with fixed seeds; the semi-synthetic and bandit experiments state their generative parameters in Appendix~A, and theory-versus-simulation agreement is reported alongside each experiment. Supplementary material is also available via \url{https://anonymous.4open.science/r/Signed_Exposure}.
\end{document}

%% file: math_commands.tex
\usepackage{amsmath,amsfonts,bm}

\def\eqref#1{equation~\ref{#1}}

\def\1{\bm{1}}

\DeclareMathAlphabet{\mathsfit}{\encodingdefault}{\sfdefault}{m}{sl}
\SetMathAlphabet{\mathsfit}{bold}{\encodingdefault}{\sfdefault}{bx}{n}

\newcommand{\E}{\mathbb{E}}

